\documentclass[9pt,twocolumn]{IEEEtran}
\usepackage{amsmath}
\usepackage{amssymb}
\usepackage{amsthm}
\usepackage{graphicx}
\usepackage{epstopdf}
\usepackage{bm}
\usepackage{cite}
\usepackage{subfigure}

\newtheorem{theorem}{Theorem}[section]
\newtheorem{lemma}[theorem]{Lemma}
\newtheorem{definition}{Definition}
\newtheorem{remark}[theorem]{Remark}
\newtheorem{example}[theorem]{Example}

\begin{document}

%\begin{frontmatter}

\title{Kalman Filtering With Relays Over Wireless Fading Channels}

%\author[Melbourne]{Alex S. Leong}\ead{asleong@unimelb.edu.au},   
%\author[Newcastle]{Daniel E. Quevedo}\ead{dquevedo@ieee.org}

%\address[Melbourne]{Department of Electrical and Electronic Engineering, 
%  University of Melbourne, Parkville, Vic. 3010, Australia}  % Please supply                                              
%\address[Newcastle]{School of Electrical Engineering and Computer Science, University of Newcastle, %Callaghan, NSW 2308, Australia}            

%\begin{keyword}                           
%Fading channels, Kalman filtering, packet drops, power control, relays            
%\end{keyword}  

\author{Alex S. Leong and Daniel E. Quevedo
\thanks{This work was supported by the Australian Research Council. A preliminary version of parts of this note was presented in \cite{LeongQuevedo_ACC}.}
\thanks{A. Leong is with the Department of Electrical and Electronic Engineering, 
  University of Melbourne, Parkville, Vic. 3010, Australia, 
        {\tt\small asleong@unimelb.edu.au}. D. Quevedo is with the Faculty of Electrical Engineering and Information Technology (EIM-E), University of Paderborn, Paderborn, Germany,  {\tt\small dquevedo@ieee.org}.}
}

\maketitle

\begin{abstract} 
This note studies the use of  relays to improve the performance of Kalman filtering over packet dropping links. Packet reception probabilities are governed by time-varying fading channel gains, and the sensor and relay transmit powers. We consider  situations with multiple sensors and relays, where each relay can either forward one of the sensors' measurements to the gateway/fusion center, or perform a simple linear network coding operation on some of the sensor measurements. Using an expected error covariance performance measure, we consider optimal and suboptimal methods for finding the best relay configuration, and power control problems for optimizing the Kalman filter performance. Our methods show that significant performance gains can be obtained through the use of relays, network coding and power control, with at least 30-40$\%$ less power consumption for a given expected error covariance specification. 
\end{abstract}

%\end{frontmatter}

\section{Introduction}
Since the seminal work of \cite{Sinopoli}, which showed the existence of a critical threshold on the packet arrival rate for stability  of the Kalman filter with i.i.d. Bernoulli packet losses, the problem of Kalman filtering over unreliable channels has received considerable attention. Extensions of \cite{Sinopoli} include the case of multiple sensors \cite{LiuGoldsmith,QuevedoAhlenJohansson}, further characterizations of the critical threshold \cite{MoSinopoli}, probabilistic notions of performance \cite{ShiEpsteinMurray}, the consideration of delays \cite{Schenato}, correlated packet losses \cite{HuangDey,QuevedoAhlenJohansson}, and power control for stability \cite{Quevedo_automatica}. 

In digital communications, channel coding is often used to improve the quality of transmissions over unreliable channels. The concept of network coding \cite{LiYeungCai,HoLun}, where in a network with many nodes, throughput can be increased by allowing intermediate nodes to perform simple operations (such as linear transformations \cite{LiYeungCai}) on its received information, has attracted much attention in recent years. Kalman filtering with power control and coding was considered in \cite{QuevedoAhlenOstergaard,QuevedoOstergaardAhlen}. The recent work \cite{QuevedoOstergaardAhlen} 
 included a study of network coding, where one could choose to utilize a relay to perform a network coding operation, and the energy tradeoffs involved. The use of relays in combating the effects of fading and increasing channel capacity has been extensively studied in wireless communications, see e.g. \cite{LanemanTseWornell,HostMadsenZhang}. Indeed, cooperative communications via the use of relays has been identified as one of the key enabling technologies for  fifth generation (5G) mobile networks \cite{RappaportHeath}. The use of a relay in control has been studied in \cite{KumarGuptaLaneman_CDC}, which showed that for the case of  a single sensor the stability region for stabilizing an unstable LTI plant can be increased in some situations, and also in applications towards control of unmanned aerial vehicles (UAVS) \cite{JohansenZolich}. 

In this note, we study remote estimation using relay nodes, and investigate what information the individual relays should send to the gateway/fusion center. In a related setup considered in \cite{QuevedoOstergaardAhlen}, the relay could only perform network coding  that linearly combines two of the sensor transmissions using an XOR operation \cite{LiYeungCai}. Here, we allow for the possibility of the relay combining multiple sensor transmissions using XOR operations \cite{Katti_XOR}, as well as the possibility of the relay  forwarding the sensors' transmissions, which can give better performance in certain situations.
The current work extends our recent contribution \cite{LeongQuevedo_ACC}, which  considered the situation with two sensors and one relay, and i.i.d. packet dropping links. In this note we generalize \cite{LeongQuevedo_ACC} to  multiple sensors and relays, and  additionally consider  time-varying packet reception probabilities governed by time-varying fading channel gains and the transmission power over this channel. 
%For a given power allocation, we consider optimal and suboptimal methods for finding the relay configuration that, at each time instant, minimizes the one step ahead expected error covariance. We then consider the more difficult problem of jointly optimizing the relay configuration and sensor and relay transmit powers. 

%The remainder of this manuscript is organized as follows.  The system model is described in Section \ref{model_sec}, with the Kalman filtering equations given in Section \ref{KF_sec}, and some preliminary results on the performance of the Kalman filter  in Section \ref{performance_sec}. For a given power allocation, optimal and suboptimal relay configuration selection methods are studied in Section \ref{optimal_relay_sec}. A special case is considered in Section \ref{special_case_sec}. The more difficult problem of jointly optimizing the relay configuration and sensor and relay transmit powers is investigated  in Section \ref{power_control_sec}.  Numerical studies are presented in Section \ref{numerical_sec}. Section \ref{conclusion_sec} draws conclusions. Technical proofs are included in the appendix. 

\section{System model}
\label{model_sec}
Throughout this note, $k$ will denote the discrete time index, $i$  will be used to denote the sensor indices, and $l$ the relay index. 
The process is a discrete time linear system
\begin{equation}
\label{state_eqn}
x_{k+1}=Ax_k+w_k
\end{equation}
where $x_k \in \mathbb{R}^n$ and $\{w_k\}$ is i.i.d. Gaussian with zero mean and covariance matrix $Q > 0$.\footnote{We say that a matrix $X > 0$, if it is positive definite, and $X \geq 0 $, if it is positive semi-definite.}  The process is observed by $M$ sensors with measurements
\begin{equation}
\label{measurement_eqn}
y_{i,k}= C_i x_k + v_{i,k}, \quad i=1,\dots,M
\end{equation}
with $y_{i,k} \in \mathbb{R}, \forall i$
 and $\{v_{i,k}\}$ are i.i.d. Gaussian with zero mean and variance $R_i\geq 0$, $i=1,\dots,M$. The processes $\{v_{i,k}\}$ and $\{w_k\}$ are assumed to be mutually independent, with $(A,C)$  detectable and $(A,Q^{1/2})$ stabilizable, where $C\triangleq \textrm{col}(C_1,\dots,C_M)$. 
 
 We assume that the measurements $y_{i,k}$ have undergone source coding and can be grouped into packets of $b$ bits, with each packet short enough to be transmitted within one time step. In particular, the uniform quantizer of \cite{HuiNeuhoff} will be used here. Under the additive noise model for quantization (which in general is quite accurate for bit rates as low as three bits per sample \cite{GrayNeuhoff}), the quantized value of $y_{i,k}$ can be written as
 $$y_{i,k}^q = y_{i,k} + q_{i,k}$$
 where the quantization noise $q_{i,k}$ has variance $\delta_b \mathbb{E}[y_{i,k}^2]$, with $$\delta_b = \frac{4 b \ln 2}{3 \times 2^{2b}}$$ when using the uniform quantizer of \cite{HuiNeuhoff}. 
The measurements $y_{i,k}^q$ are transmitted over orthogonal/parallel  channels to a gateway, which will perform the Kalman filtering operation. Let $\gamma_{i,k}$, for $ i=1,\dots,M$, be random variables such that $\gamma_{i,k}=1$ if $y_{i,k}^q$ is successfully transmitted to the gateway by sensor $i$, and $\gamma_{i,k}=0$ otherwise.

Furthermore,  there exist $L$ intermediate relay nodes  that can be used to aid the transmission of the sensor measurements to the gateway. Such  situations can for instance occur in mesh networks, where nodes close to the process will make measurements of the process, while the other nodes don't make measurements but can be used to relay the sensor measurements to the gateway \cite{LanemanTseWornell}. A diagram of the system model for the case of $M=2$ sensors and $L=2$ relays is given in Fig. \ref{system_model_two_relay_new}.
\begin{figure}[htb!]
\centering 
\includegraphics[scale=0.28]{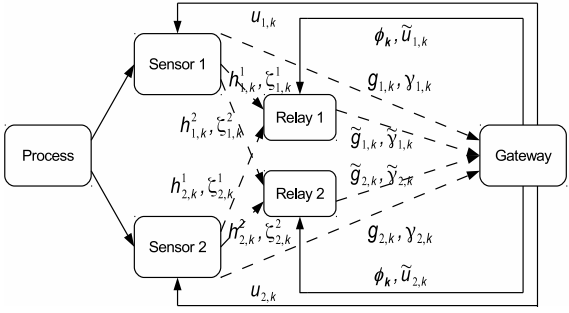} 
\caption{System model for the case of two sensors and two relays}
\label{system_model_two_relay_new}
\end{figure} 
Each relay can listen to a subset of the sensor transmissions. Denote $\mathcal{I} = \{1,\dots,M\}$ as the  set of all sensors, and  $\mathcal{I}_l \subseteq \mathcal{I}$ as the set of sensors which relay $l$ can listen to. In general the sets $\mathcal{I}_l, l=1,\dots,L$ will not necessarily be disjoint, with possibly multiple relays listening to a given sensor.  For $i \in \mathcal{I}_l$, let $\zeta_{i,k}^l$ be a random variable such that $\zeta_{i,k}^l=1$ if the transmission at time $k$ of sensor $i$ is received by relay $l$, and $\zeta_{i,k}^l=0$ otherwise. 
The relays can perform some simple local processing before transmitting over orthogonal channels to the gateway. Let $\tilde{\gamma}_{l,k}$ for $l=1,\dots,L$ be  random variables such that $\tilde{\gamma}_{l,k}=1$ if transmission at time $k$ from  relay $l$ to the gateway is successful, and $\tilde{\gamma}_{l,k}=0$ otherwise. 

In this note we will consider a few simple operations that the relay can perform.\footnote{We assume limited computational power at the relays, thus only simple operations at a bit-level are considered in this note. If, however, additional computational capability is available, then other possibilities include the use of more involved network coding schemes \cite{HoLun} or the computation of local state estimates at the relays \cite{XuHespanha}.} A relay can either: 
a) listen to one of the sensors' transmissions, say sensor $i$, and  forward $y_{i,k}^q$ if it is successfully received by the relay, or 
b) listen to a number  of the sensors' transmissions, say sensors $i_1, i_2,\dots,i_l$, and send $y_{i_1,k}^q \oplus y_{i_2,k}^q \oplus \dots \oplus y_{i_l,k}^q$ if $y_{i_1,k}^q, y_{i_2,k}^q, \dots y_{i_l,k}^q $ have all been successfully received by the relay, where $\oplus$ is the XOR operation. The XOR operation is commonly used in network coding \cite{LiYeungCai,Katti_XOR}. For instance, if the gateway receives both $y_{i,k}^q$ and $y_{i,k}^q \oplus y_{j,k}^q$, or both $y_{j,k}^q$ and $y_{i,k}^q \oplus y_{j,k}^q$, then the gateway can recover both $y_{i,k}^q$ and $y_{j,k}^q$ with a binary subtraction. In general, given the transmissions received at the gateway, the measurements which can be recovered can be determined using Gaussian elimination over $\mathbb{Z}_2$.\footnote{While our use of the XOR operation is similar to network coding, our objectives are not exactly the same. In network coding transmissions are often regarded as ``successful'' only if all packets arrive (eventually) at their intended destinations, whereas in our problem even if some packets are not received one still can perform state estimation using the available measurements.}  Determining which sensor/s each relay listens to, and which operation each relay uses, is one of the key questions addressed in this work, see Section \ref{optimal_relay_sec}. We define a \emph{relay configuration} 
$$\bm{\phi}_k = (\phi_{1,k},\dots,\phi_{L,k})$$ 
at time $k$ as the set of operations $\phi_{l,k}$ that each relay uses at time $k$. The set of all possible relay configurations  will be denoted by $\bm{\Phi}$.

The communication channels will be modelled as time-varying fading channels \cite{Proakis}. We let $g_{i,k}, i=1,\dots,M$ be  the channel gains at time $k$ from the sensor $i$ to the gateway, $\tilde{g}_{l,k}, l=1,\dots,L$ the channels gains from  relay $l$ to the gateway, and $h_{i,k}^l, i \in \mathcal{I}_l, l=1,\dots,M$ the channel gains from  sensor $i$ to relay $l$. We use the block fading model \cite{CaireTariccoBiglieri} and assume that $\{g_{i,k}\}, \{\tilde{g}_{l,k}\}, \{h_{i,k}^l\}$ vary over time $k$ in an i.i.d. manner, with the processes being mutually independent.   Denote the transmit powers at time $k$ of the sensors and relays   by $u_{i,k}, i=1,\dots,M$ and  $\tilde{u}_{l,k}, l=1,\dots,L$ respectively. Following the model of \cite{Quevedo_automatica}, the packet reception probabilities will depend on both the channel gains and transmit powers as follows: We have 
$\lambda_{i,k} \triangleq \mathbb{P}(\gamma_{i,k}=1|g_{i,k},u_{i,k}) = p(g_{i,k} u_{i,k})$ as the time-varying\footnote{This differs from the i.i.d. Bernoulli packet drop model with constant packet reception probabilities considered in e.g.\cite{Sinopoli}} packet reception probabilities from sensor $i$ to the gateway, $ \tilde{\lambda}_{l,k} \triangleq \mathbb{P}(\tilde{\gamma}_{l,k}=1|\tilde{g}_{l,k}, \tilde{u}_{l,k}) = p(\tilde{g}_{l,k} \tilde{u}_{l,k})$ the  probabilities from relay $l$ to the gateway, and 
$ \rho_{i,k}^l \triangleq \mathbb{P}(\zeta_{i,k}^l=1|h_{i,k}^l, u_{i,k}) = p(h_{i,k}^l u_{i,k})$  the  probabilities from sensor $i$  to relay $l$. Here $p(.): [0,\infty) \rightarrow [0,1]$ is a continuous  monotonically increasing function whose form depends on the particular digital  modulation and coding scheme being used \cite{Proakis}. For example, in the case of uncoded binary phase shift keying (BPSK) transmission  with $b$ bits per packet,  $p(.)$  would take  the form
\begin{equation}
\label{f_BPSK}
p(g u) = \left( \int_{-\infty}^{\sqrt{g u}} \frac{1}{\sqrt{2 \pi}} e^{-t^2/2} dt \right)^b
\end{equation}
where we assume a packet is successfully received if all $b$ bits are succesfully received. However, if there is channel coding and/or different digital modulation schemes, $p(.)$ will take on  different forms \cite{GatsisRibeiroPappas}. 
In Table \ref{notation_table} we summarize the notation for the channel gains, packet reception random variables and packet reception probabilities for the different types of links. 
\begin{table}
\caption{Notation for different types of  links}
\centering
\begin{tabular}{|c|c|c|l|} \hline
& Channel Gain &  Packet Reception  &  Packet Reception \\ 
& & Random Variable & Probability \\ \hline
Sensor $i$ to Gateway & $g_{i,k}$ & $\gamma_{i,k}$ & $\quad\quad \lambda_{i,k} $ \\ \hline
Relay $l$ to Gateway & $\tilde{g}_{i,k}$ & $\tilde{\gamma}_{i,k}$ &  $\quad\quad \tilde{\lambda}_{i,k} \phantom{\tilde{\tilde{\lambda}}_{i,k}}$ \\ \hline
Sensor $i$ to Relay $l$ & $h_{i,k}^l $ & $ \zeta_{i,k}^l$ & $\quad\quad \rho_{i,k}^l$ \\ \hline
\end{tabular}
\label{notation_table}
\end{table}

In addition, there are feedback links from the gateway to the sensors and relays which can be used to communicate the relay configuration $\bm{\phi}_k$ and power levels $u_{i,k}$ and $\tilde{u}_{l,k}$ to be used, see Section \ref{optimal_relay_sec} and Section \ref{power_control_sec}. 
In this note we will assume that transmissions can occur over a much faster time scale than the process (\ref{state_eqn}). Thus, delays experienced by the measurements in passing through intermediate relay nodes will be ignored. For instance, in the industrial wireless sensor networks standard WirelessHART \cite{WirelessHART_control}, transmissions between nodes would typically take around 10 ms, whereas for many estimation and control applications the process time constant might be 1 sec or more.

\begin{table}
\caption{Boolean expressions for  $\theta_{1,k}$ and $\theta_{2,k}$, for two sensors and one relay, under three types of operations}
\centering
\begin{tabular}{|c|c|c|c|} \hline
%&  \multicolumn{3}{c|}{Boolean expression under scheme} \\  \cline{2-4}
$\bm{\phi}_k$ & Forward $y_{1,k}^q$ & Forward $y_{2,k}^q$ & Send $y_{1,k}^q \oplus y_{2,k}^q$\\ \hline 
$\!\!\theta_{1,k}\!\!$ & $\!\! \gamma_{1,k} \!\vee\! (\tilde{\gamma}_{1,k} \!\wedge\! \zeta_{1,k}^1)\!\!$ & $ \gamma_{1,k}$ &  $\!\!\gamma_{1,k} \!\vee\! ( \tilde{\gamma}_{1,k} \!\wedge\! \gamma_{2,k} \!\wedge\! \zeta_{1,k}^1 \!\wedge\! \zeta_{2,k}^1 )\!\!$ \\ \hline
$\!\!\theta_{2,k}\!\!$ & $\gamma_{2,k}$ & $\!\! \gamma_{2,k} \!\vee\! (\tilde{\gamma}_{1,k} \!\wedge\! \zeta_{2,k}^1) \!\!$ &  $\!\!\gamma_{2,k} \!\vee\! (\tilde{\gamma}_{1,k} \!\wedge\! \gamma_{1,k} \!\wedge\! \zeta_{1,k}^1 \!\wedge\! \zeta_{2,k}^1 )\!\!$ \\ \hline
\end{tabular}
\label{truth_table}
\end{table}

\section{Kalman filter with packet drops and relays}
\label{KF_sec}
Let $\theta_{i,k}, i=1,\dots,M$ be random variables such that $\theta_{i,k}=1$ if $y_{i,k}^q$ can be reconstructed at the gateway, and $\theta_{i,k}=0$ otherwise.  Note that in general $\theta_{i,k}$ and $\theta_{j,k}, i \neq j$ are not independent. 
Values of $\theta_{i,k}$ for different relay configurations and combinations of $\gamma_{i,k}$,  $\tilde{\gamma}_{l,k}$, $\zeta_{i,k}^l$, can be written in Boolean algebra form. 
For example, in Table \ref{truth_table} we give the Boolean expressions for $\theta_{1,k}$ and $\theta_{2,k}$ in the case of two sensors and one relay, where we use the notation $\wedge$ to denote logical `and' and $\vee$ to denote logical `or'. 
Now define
%\begin{equation}
%\label{breve_C_defn}
% \breve{C}_k \triangleq \left[ \begin{array}{c} \theta_{1,k} C_1 \\ \vdots \\ \theta_{M,k} C_M \end{array} \right], \breve{y}_k  \triangleq \left[ \begin{array}{c} \theta_{1,k} y_{1,k} \\ \vdots \\ \theta_{M,k} y_{M,k} \end{array} \right] 
% \end{equation}
 \begin{equation}
\label{breve_C_defn}
\begin{split}
& \breve{C}_k \triangleq \textrm{col}(\theta_{1,k} C_1,\dots,\theta_{M,k} C_M), \, \breve{y}_k \triangleq \textrm{col}(\theta_{1,k} y_{1,k}^q,\dots, \theta_{M,k} y_{M,k}^q), \\
&\hat{x}_{k+1|k}  \triangleq  \mathbb{E} [x_{k+1} | \breve{y}_0,\dots, \breve{y}_k, \breve{C}_0,\dots,\breve{C}_k] \\
&P_{k+1|k}  \!\triangleq  \!\mathbb{E} [(x_{k+1} \!-\! \hat{x}_{k+1|k})(x_{k+1} \!- \! \hat{x}_{k+1|k})^T \!| 
%\\ & \quad\quad\quad  
\breve{y}_0,\dots, \breve{y}_k, \breve{C}_0,\dots,\breve{C}_k].
\end{split}
\end{equation}
The associated Kalman filter equations which are run at the gateway can be written as (see e.g. \cite{QuevedoAhlenOstergaard})
\begin{equation}
\label{KF_eqns}
\begin{split}
\hat{x}_{k+1|k} & = A \hat{x}_{k|k-1} + K_k (\breve{y}_k - \breve{C}_k \hat{x}_{k|k-1}) \\
P_{k+1|k} & = A P_{k|k-1} A^T + Q - K_k \breve{C}_k P_{k|k-1} A^T 
\end{split}
\end{equation}
where $K_k = A P_{k|k-1} \breve{C}_k^T (\breve{C}_k P_{k|k-1} \breve{C}_k^T + \breve{R}_k)^{-1}$, with $\breve{R}_k = \textrm{diag}(\breve{R}_{1,k},\dots,\breve{R}_{M,k}) \triangleq \textrm{diag}(R_1+\delta_b \mathbb{E} [y_{1,k}^2],\dots,R_M+\delta_b \mathbb{E} [y_{M,k}^2])$, similar to \cite{QuevedoOstergaardAhlen}. 
In the sequel, we will also call $P_{k} \triangleq P_{k|k-1}$. 

\begin{remark}
As in \cite{QuevedoOstergaardAhlen}, the Kalman filter (\ref{KF_eqns}) uses all successfully reconstructed measurements, but those measurements where $\theta_{i,k}=0$ are not taken into account. 			
\end{remark}

\section{Performance of the Kalman filter with relays}
\label{performance_sec}
In Section \ref{model_sec} we have proposed that each relay can either listen to transmissions from one of the sensors which it then forwards to the gateway, or it listens to a number of sensors and performs an XOR  operation that is then sent to the gateway. We wish to investigate which operation each relay should use, and which sensors each relay should listen to, i.e. determining the relay configuration $\bm{\phi}_k$, in order to give the best performance for the Kalman filter. This section presents some preliminary results on the performance of the Kalman filter, with optimal relay configuration selection to be studied in Section \ref{optimal_relay_sec}.
For time-varying Kalman filters with packet drops, one commonly used performance measure is the expected error covariance.  In this note we consider the problem of optimal relay configuration selection in order to minimize the one step ahead expected error covariance $\mathbb{E}[P_{k+1}|P_k,\textbf{g}_k,\bm{\phi}_k]$, where 
\begin{equation}
\label{gk_vector}
\textbf{g}_k \triangleq \{g_{1,k},\dots,g_{M,k},\tilde{g}_{1,k},\dots,\tilde{g}_{L,k},h_{1,k}^1,\dots,h_{M,k}^L  \}
\end{equation}
 represents the channel gains at time $k$, which in turn will determine the packet reception probabilities $\lambda_{i,k}$,  $\tilde{\lambda}_{l,k}$, $\rho_{i,k}^l, i=1,\dots,M,l=1,\dots,L$. 
In order to compute $\mathbb{E}[P_{k+1}|P_k,\textbf{g}_k,\bm{\phi}_k]$, we will further assume that full channel state information (CSI) at the receiver is available,  so that $\textbf{g}_k$ is known at the gateway.\footnote{In practice this can be achieved using channel estimation and prediction algorithms,  see references in \cite{Quevedo_automatica,QuevedoOstergaardAhlen}.}
%\subsection{Preliminary results}
Define 
\begin{equation}
\label{f_defn}
\begin{split}
\!f_k(X) \!& \triangleq\! A X A^T \!\!+\! Q \!-\!  \mathbb{E}\left[ A X \breve{C}_k^T  \left(  \breve{C}_k X \breve{C}_k^T \!\!+\! \breve{R}_k \right)^{\!\!-1} \!\!\breve{C}_k X A^T \Big|\textbf{g}_k, \bm{\phi}_k \right]
\end{split}
\end{equation}
where the expectation is with respect to $\theta_{1,k},\dots,\theta_{M,k}$ in the definition of $\breve{C}_k$ in (\ref{breve_C_defn}). Equivalently, we can write $f_k(X)$ as:
\begin{equation}
\label{f_defn_alternate}
\begin{split}
\!f_k(X) \!& =\! A X A^T \!\!+\! Q \!-\! \sum_{I \subseteq \mathcal{I}} \! \mathbb{E}\Big[ \! \prod_{i \in I} \Theta_{i,k} \! \prod_{j \notin I} (1- \Theta_{j,k}) \big|\textbf{g}_k, \bm{\phi}_k \!\Big]\! A X \bar{C}(I)^T 
\\ & \quad \times 
\left(  \bar{C}(I) X \bar{C}(I)^T \!\!+\! \bar{R}_k(I) \right)^{\!\!-1} \!\!\bar{C}(I) X A^T 
\end{split}
\end{equation}
where $\bar{C}(I) = \textrm{col}(\{C_i, i \in I\})$, $\bar{R}_k(I) = \textrm{diag}(\{\breve{R}_{i,k}, i \in I\})$, and $\Theta_{i,k}, i=1,\dots,M$  are random variables with the same distributions as $\theta_{i,k}$. 
%Both forms (\ref{f_defn}) and (\ref{f_defn_alternate}) will be used in this note. 
The quantities $\mathbb{E}\Big[\prod_{i \in I} \Theta_{i,k} \prod_{j \notin I} (1- \Theta_{j,k}) \big|\textbf{g}_k, \bm{\phi}_k \Big] $ can be computed in terms of the packet reception probabilities $\lambda_{i,k}$,  $\tilde{\lambda}_{l,k}$, $\rho_{i,k}^l, i=1,\dots,M,l=1,\dots,L$. A systematic procedure for doing this is as follows:
\\1)
 Write out the Boolean expression 
\begin{equation}
\label{Boolean_expression}
\bigwedge_{i \in I} \theta_{i,k} \bigwedge_{j \notin I} (\neg  \theta_{j,k})
\end{equation}
where each $\theta_{i,k}$ is written as a Boolean expression,  $\neg \theta_{j,k}$ denotes the negation of the Boolean expression for $\theta_{j,k}$, and the notation $\bigwedge_{i \in I} \theta_{i,k} \triangleq \theta_{i_1,k} \wedge \dots \wedge \theta_{i_n,k} \wedge \dots$ for indices $i_n \in I$. 
\\2) Convert the Boolean expression (\ref{Boolean_expression}) into the sum of products normal form \cite{DallyHarting}. (Note that this can be done systematically.)
\\3) $\mathbb{E}\Big[\prod_{i \in I} \Theta_{i,k} \prod_{j \notin I} (1- \Theta_{j,k}) |\textbf{g}_k, \bm{\phi}_k \Big] $ is then given by taking the sum of products normal form of (\ref{Boolean_expression}), and replacing 
$\wedge$ with multiplication, $\vee$ with addition, $\gamma_{i,k}$ with $\lambda_{i,k}$, $\neg \gamma_{i,k}$ with $1-\lambda_{i,k}$, $\tilde{\gamma}_{i,k}$ with $\tilde{\lambda}_{i,k}$,  $\neg \tilde{\gamma}_{i,k}$ with $1-\tilde{\lambda}_{i,k}$, $\zeta_{i,k}^l$ with $\rho_{i,k}^l$, and $\neg \zeta_{i,k}^l$ with $1-\rho_{i,k}^l$.

By step 2), each term in the sum will correspond to a distinct entry of the truth table for $\theta_{1,k},\dots,\theta_{M,k}$, thus allowing   $\mathbb{E}\Big[\prod_{i \in I} \Theta_{i,k} \prod_{j \notin I} (1- \Theta_{j,k}) |\textbf{g}_k, \bm{\phi}_k \Big] $ to be easily calculated. 

We now give a result on how the packet reception probabilities affect the  expected error covariance $\mathbb{E}[P_{k+1}|P_k, \textbf{g}_k, \bm{\phi}_k] = f_k(P_k)$. 
First denote: 
\begin{equation*}
\begin{split}
 \mathcal{L}_{i,k} &\triangleq \{\lambda_{1,k},\dots,\lambda_{i,k},\dots,\lambda_{M,k},\tilde{\lambda}_{1,k},\dots,\tilde{\lambda}_{L,k},\rho_{1,k}^1,\dots,\rho_{M,k}^L\} \\
 \mathcal{U}_{i,k} & \triangleq \{\lambda_{1,k},\dots,\mu_{i,k},\dots,\lambda_{M,k},\tilde{\lambda}_{1,k},\dots,\tilde{\lambda}_{L,k},\rho_{1,k}^1,\dots,\rho_{M,k}^L \} 
 \end{split}
 \end{equation*}
\begin{equation*}
\begin{split}
 \tilde{\mathcal{L}}_{l,k} & \triangleq \{ \lambda_{1,k},\dots,\lambda_{M,k},\tilde{\lambda}_{1,k},\dots,\tilde{\lambda}_{l,k},\dots,\tilde{\lambda}_{L,k},\rho_{1,k}^1,\dots,\rho_{M,k}^L \} \\
 \tilde{\mathcal{U}}_{l,k} & \triangleq \{\lambda_{1,k},\dots,\lambda_{M,k},\tilde{\lambda}_{1,k},\dots,\tilde{\mu}_{l,k},\dots,\tilde{\lambda}_{L,k},\rho_{1,k}^1,\dots,\rho_{M,k}^L \} 
 \end{split}
 \end{equation*}
 \begin{equation*}
 \begin{split}
 \mathcal{R}_{i,k}^l & \triangleq \{ \lambda_{1,k},\dots,\lambda_{M,k},\tilde{\lambda}_{1,k},\dots,\tilde{\lambda}_{L,k},\rho_{1,k}^1,\dots,\rho_{i,k}^l,\dots,\rho_{M,k}^L \} \\
 \mathcal{S}_{i,k}^l & \triangleq \{ \lambda_{1,k},\dots,\lambda_{M,k},\tilde{\lambda}_{1,k},\dots,\tilde{\lambda}_{L,k},\rho_{1,k}^1,\dots,\sigma_{i,k}^l,\dots,\rho_{M,k}^L  \}
\end{split}
\end{equation*}

\begin{lemma}
\label{lambda_monotonicity_lemma}
Let $f_{\mathcal{X}_{i,k}} (.)$ be defined by $f_k(.)$ in (\ref{f_defn}) when the links have packet reception probabilities $\mathcal{X}_{i,k}$. Then, irrespective of which relay configuration is used,  $\forall i=1,\dots,M, \forall l=1,\dots,L$,
\begin{equation*}
\begin{split}
\lambda_{i,k} \leq \mu_{i,k} & \Rightarrow f_{\mathcal{L}_{i,k}}(X)  \geq  f_{\mathcal{U}_{i,k}}(X)  \\
\tilde{\lambda}_{l,k} \leq \tilde{\mu}_{l,k} & \Rightarrow f_{\tilde{\mathcal{L}}_{l,k}} (X) \geq  f_{\tilde{\mathcal{U}}_{l,k}} (X) \\
\rho_{i,k}^l \leq \sigma_{i,k}^l & \Rightarrow f_{\mathcal{R}_{i,k}^l}(X) \geq  f_{\mathcal{S}_{i,k}^l}(X)
\end{split}
\end{equation*}
\end{lemma}
%\begin{proof}
%See the appendix.
%\end{proof}

\begin{proof}
Consider the case $\lambda_{i,k} \leq \mu_{i,k}$. Recall that Bernoulli random variables can be generated from $U(0,1)$ uniform random variables, by comparing the uniform random variable with the probability that the Bernoulli random variable is equal to one, i.e.. 
$\gamma_{i,k} = 1$ when $u \leq \lambda_{i,k}$, and $\gamma_{i,k}=0$ otherwise, 
where $u$ is $U(0,1)$. Let $\omega$ denote an outcome corresponding to $N$ independent realizations of $U(0,1)$ random variables, where $N$ is equal to the total number of packet dropping links. For each $\omega$, one can generate corresponding independent Bernoulli random variables $\gamma_{1,k}, \dots,\gamma_{M,k}, \tilde{\gamma}_{1,k},\dots,\tilde{\gamma}_{L,k}, \zeta_{1,k}^1,\dots,\zeta_{M,k}^L$. One can then construct the Bernoulli random variables $\theta_{1,k},\dots,\theta_{M,k}$, and hence $\breve{C}_k$ as in (\ref{breve_C_defn}). 

Let $\breve{C}_{\mathcal{L}_{i,k}}(\omega)$ be the matrix $\breve{C}_k$ when using packet reception probabilities $\mathcal{L}_{i,k}$ and $\breve{C}_{\mathcal{U}_{i,k}}(\omega)$ be the matrix $\breve{C}_k$ when using packet reception probabilities $\mathcal{U}_{i,k}$. Now note that if $\theta_{j,k}(\omega)=1$ using the packet reception probabilities $\mathcal{L}_{i,k}$, then we also have $\theta_{j,k}(\omega)=1$  when using the packet reception probabilities $\mathcal{U}_{i,k}$, by the way in which $\theta_{j,k}(\omega)$ is constructed, and  since an increase in the packet reception probability of any link cannot decrease the probability of reconstructing any of the sensor measurements. Hence
\begin{equation}
\label{sample_path_monotonicity}
\begin{split}
& A X \breve{C}_{\mathcal{L}_{i,k}}(\omega)^T  \left(  \breve{C}_{\mathcal{L}_{i,k}}(\omega) X \breve{C}_{\mathcal{L}_{i,k}}(\omega)^T \!\!+\! \breve{R}_k \right)^{\!\!-1} \!\!\!\breve{C}_{\mathcal{L}_{i,k}}(\omega) X A^T \\ & \geq A X \breve{C}_{\mathcal{U}_{i,k}}(\omega)^T  \left(  \breve{C}_{\mathcal{U}_{i,k}}(\omega) X \breve{C}_{\mathcal{U}_{i,k}}(\omega)^T \!\!+\! \breve{R}_k \right)^{\!\!-1} \!\!\!\breve{C}_{\mathcal{U}_{i,k}}(\omega) X A^T
\end{split}
\end{equation}
Since (\ref{sample_path_monotonicity}) holds for all $\omega$, we have
\begin{equation*}
\begin{split}
& \mathbb{E} \left[A X (\breve{C}_{\mathcal{L}_{i,k}})^T  \left(  \breve{C}_{\mathcal{L}_{i,k}} X (\breve{C}_{\mathcal{L}_{i,k}})^T \!\!+\! \breve{R}_k \right)^{\!\!-1} \!\!\!\breve{C}_{\mathcal{L}_{i,k}} X A^T \Big| \textbf{g}_k, \bm{\phi}_k \right] \\ & \quad \geq \mathbb{E}\left[A X (\breve{C}_{\mathcal{U}_{i,k}})^T  \left(  \breve{C}_{\mathcal{U}_{i,k}} X (\breve{C}_{\mathcal{U}_{i,k}})^T \!\!+\! \breve{R}_k \right)^{\!\!-1} \!\!\!\breve{C}_{\mathcal{U}_{i,k}} X A^T \Big| \textbf{g}_k, \bm{\phi}_k \right]
\end{split}
\end{equation*}
which shows that $f_{\mathcal{L}_{i,k}} (X)\geq  f_{\mathcal{U}_{i,k}}(X) $. 
The other two cases can be proved in a similar manner. 
\end{proof}

Lemma \ref{lambda_monotonicity_lemma} says that increasing the packet arrival rate on any one of the packet dropping links will give a decrease in the value of the one step ahead expected error covariance, no matter which relay configuration is used.

We also wish to determine how the system parameters such as $C_1,\dots, C_M, R_1, \dots, R_M$ will affect the values of the expected error covariance. Denote:
\begin{equation*}
\begin{split}
\mathcal{A}_{i} & \triangleq \{C_1,\dots,C_i,\dots,C_M,R_1,\dots,R_i,\dots,R_M\} \\
\mathcal{B}_{i} & \triangleq \{C_1,\dots,D_i,\dots,C_M,R_1,\dots,S_i,\dots,R_M\}
\end{split}
\end{equation*}
\begin{lemma}
\label{SNR_monotonicity_lemma}
Let $f_{\mathcal{Y}_{i},k} (X)$ be defined by $f_k(.)$ in (\ref{f_defn_alternate}) when the sensors have parameters $\mathcal{Y}_{i}$. Then, irrespective of which relay configuration is used, $\forall i=1,\dots,M$, 
\begin{equation*}
\begin{split}
& C_i^T \breve{R}_{i,k}^{-\!1} C_i \leq D_i^T \breve{S}_{i,k}^{-\!1} D_i  \Rightarrow f_{\mathcal{A}_{i}}(X)  \geq f_{\mathcal{B}_{i}}(X)
\end{split}
\end{equation*}
\end{lemma}
The proof of Lemma \ref{SNR_monotonicity_lemma} is omitted; a related result is proved in  \cite{LeongQuevedo_ACC}.
The quantity $C_i^T \breve{R}_{i,k}^{-1} C_i$ can be regarded as the signal-to-noise ratio (SNR) for sensor $i$ (which includes the quantization noise contribution). 
Thus Lemma \ref{SNR_monotonicity_lemma} shows that the expected error covariance is monotonic in the SNRs $C_i^T \breve{R}_{i,k}^{-1} C_i$.

\section{Optimal relay configuration selection}
\label{optimal_relay_sec}
We now wish  to address the question of determining which configurations for the relays will give the best Kalman filter performance.  
Suppose for now that the sensor transmit powers $u_{i,k},i=1,\dots,M$ and relay transmit powers $\tilde{u}_{l,k}, l=1,\dots,L$ are given or fixed (The more difficult problem of jointly optimizing the relay configuration and transmission powers will be considered in Section \ref{power_control_sec}.).  
We wish to choose at each time instant $k$, the relay configuration $\bm{\phi}^*_k$ that minimizes $\textrm{Tr} \mathbb{E}[P_{k+1}|P_k,  \textbf{g}_k,\bm{\phi}_k]$, i.e.
\begin{equation}
\label{optimal_relay_prob}
\bm{\phi}^*_k = \textrm{arg} \!\!\!\!\!\!\!\!\min_{\!\!\!\bm{\phi}_k(P_k,\textbf{g}_k) \in \bm{\Phi}} \! \textrm{Tr} \mathbb{E}[P_{k+1}|P_k,  \textbf{g}_k,\bm{\phi}_k].
\end{equation}
where $\mathbb{E}[P_{k+1}|P_k, \textbf{g}_k,\bm{\phi}_k] = f_k(P_k)$, see (\ref{f_defn_alternate}).

\subsection{Optimal relay configuration selection}
\label{exhaustive_search_sec}
 Problem (\ref{optimal_relay_prob}) can, in principle, be solved by exhaustive search at the gateway. The optimal configuration can then be fed back to the relays. 
 We will characterize the number of  relay configurations that need to be checked at each time instant for exhaustive search.
\begin{lemma}
\label{num_configs_lemma}
Let  $\mathcal{I}_l$ be the set of sensors that relay $l$ can listen to, and let $M_l = | \mathcal{I}_l| $ denote the cardinality of $\mathcal{I}_l$. Suppose that there are no restrictions on how many relays listen to the same sensor. 
Then the number of possible relay  configurations for $\bm{\phi}_k$ is
\begin{equation}
\label{num_configs}
|\bm{\Phi}| = \prod_{l=1}^L \left(2^{M_l} - 1 \right)
\end{equation}
\end{lemma}
\begin{proof}
%See the appendix.
First fix a relay $l$, which can listen to $M_l$ of the sensors. This relay can either forward any one of the sensor transmissions, or perform the XOR operation on two or more of the sensor transmissions it listens to, resulting in $M_l +  \binom{M_l}{2} + \binom{M_l}{3} + \dots + \binom{M_l}{M_l} = 2^{M_l}-1 $ possible operations. If there are no restrictions on multiple relays listening to the same sensor, then by the multiplication principle  the number of relay  configurations is
$\prod_{l=1}^L \left( 2^{M_l}-1 \right)$.
\end{proof}
 We thus see that the number of configurations that needs to be checked  is in the worst case (where each relay can listen to all sensors)  exponential in the number of relays $L$ and number of sensors $M$. However, in practice, due to geographical considerations, the number of sensors $M_l$ that each sensor $l$ listens to is often small, e.g. in \cite{TangWang_VTC,RajalinghamHoLeNgoc} it is assumed that $M_l \leq 3$. 
%Moreover, the solution to problem (\ref{optimal_relay_prob}) can also be obtained by solving a linear program, which will be presented in the next subsection. 

\subsection{Stability of Kalman filtering with relay configuration selection}
We now wish to give a condition for stability of the Kalman filter with optimal relay configuration selection. 
\begin{definition}
The Kalman filter is said to be exponentially bounded if there exist finite constants $\alpha$ and $\beta$, and $\rho \in [0,1)$, such that 
$\textnormal{Tr}\mathbb{E}[ P_k] \leq \alpha \rho^k + \beta, \quad \forall k $.
\end{definition}

\begin{theorem}
\label{stability_thm}
Let  $\{s_k\}$ be a stochastic process such that 
$s_k = 1 $ if $\breve{C}_k$ is full rank, and $s_k= 0 $ otherwise. 
Suppose  there exists a  policy $\bm{\phi}^\sharp(\textbf{g}_k)$  dependent only on $\textbf{g}_k$, such that 
\begin{equation}
\label{stability_condition}
||A||^{2}  \mathbb{E}[\mathbb{P}(s_k=0|\textbf{g}_k, \bm{\phi}^\sharp (\textbf{g}_k))] < 1, 
\end{equation}
where $||A||$ is the spectral norm of $A$. Then 
the Kalman filter  using the  configurations obtained from problem (\ref{optimal_relay_prob})  is exponentially bounded. 
\end{theorem}
\begin{proof}
Since the distribution of $\breve{C}_k$ depends on $P_k$ and $\textbf{g}_k$, and $\textbf{g}_k$ is independent in time and of $P_k$, we have
$\mathbb{P} (\breve{C}_k|P_k,P_{k-1},\dots,P_0) = \mathbb{P} (\breve{C}_k|P_k) $.
Then by (\ref{KF_eqns}), the process $\{P_k\}$ is Markovian.
Now define $V_k \triangleq \textrm{Tr} P_k$.
 Since in the relay configuration selection problem (\ref{optimal_relay_prob}) we are minimizing $\mathbb{E} \{ V_{k+1} | P_k, \textbf{g}_k, \phi_k(P_k,\textbf{g}_k) \}$,  we have
\begin{equation*}
\begin{split}
& \mathbb{E} \{ V_{k+1} | P_k\}  = \mathbb{E}[\mathbb{E} \{ V_{k+1} | P_k, \textbf{g}_k, \bm{\phi}^*_k(P_k,\textbf{g}_k)\}] \\ & \leq \mathbb{E}[\mathbb{E} \{ V_{k+1} | P_k, \textbf{g}_k, \bm{\phi}^\sharp(\textbf{g}_k)\} ]\\
& = \mathbb{E}[ \mathbb{E} \{ V_{k+1} | P_k, \textbf{g}_k, \bm{\phi}^\sharp(\textbf{g}_k), s_k=1\} \mathbb{P} \{ s_k=1|P_k, \textbf{g}_k, \bm{\phi}^\sharp(\textbf{g}_k)\} \\ &\quad +  \mathbb{E} \{ V_{k+1} | P_k, \textbf{g}_k, \bm{\phi}^\sharp(\textbf{g}_k), s_k=0\} \mathbb{P} \{ s_k=0|P_k, \textbf{g}_k, \bm{\phi}^\sharp(\textbf{g}_k)\} ]  \\
& \leq  W + \left( ||A||^2 V_k + \textrm{Tr} Q \right) \mathbb{E}[\mathbb{P} \{ s_k=0| \textbf{g}_k, \bm{\phi}^\sharp(\textbf{g}_k)\}]
\end{split}
\end{equation*}
where the last inequality is shown using similar arguments to \cite{QuevedoOstergaardAhlen,Quevedo_automatica}, and $W$ is a positive constant. 
If $||A||^{2}  \mathbb{E}[\mathbb{P}(s_k=0|\textbf{g}_k, \bm{\phi}^\sharp(\textbf{g}_k)) ]< 1$ 
we may then use a stochastic Lyapunov function argument similar to \cite{QuevedoOstergaardAhlen} to show that $
\mathbb{E} \{ V_{k} | P_0\} \leq \alpha \rho^{k} + \beta, \forall k
$
for some $\rho \in [0,1)$ and constants $\alpha$ and $\beta$, which establishes exponential boundedness of the Kalman filter. 
\end{proof}
Theorem \ref{stability_thm} thus provides a sufficient condition for Kalman filter stability dependent on the system matrix $A$ and the distributions of the channel gains $\textbf{g}_k$. 

\begin{example}
Consider the case of two sensors and one relay, with  $\breve{C}_k$ being full rank only when $\theta_{1,k}=\theta_{2,k}=1$. Then 
\begin{equation*}
\begin{split}
\mathbb{P}(\theta_{1,k}=1, \theta_{2,k}=1 | \textbf{g}_k, \bm{\phi}_k) &= \mathbb{E}[\Theta_{1,k} \Theta_{2,k} | \textbf{g}_k, \bm{\phi}_k] \\
&= \lambda_{1,k} \lambda_{2,k} +(1-\lambda_{1,k})\lambda_{2,k} \tilde{\lambda}_{1,k} \rho_{1,k}^1
\end{split}
\end{equation*} 
Suppose we choose $\bm{\phi}^\sharp$ to be the suboptimal policy that always forwards $y_{1,k}^q$, and with the transmit powers $u_{1,k}=u_{2,k}=\tilde{u}_{1,k}=1$. 
The condition (\ref{stability_condition}) then becomes 
\begin{equation*}
\begin{split}
&\mathbb{E}[\mathbb{P}(s_k=0|\textbf{g}_k, \bm{\phi}^\sharp (\textbf{g}_k))]  = \int \Big(1 \!-\! p(g_{1,k}) p(g_{2,k}) \\ & \quad -\! (1\!-\!p(g_{1,k})) p(g_{2,k}) p(\tilde{g}_{1,k}) p(h_{1,k}^1) \Big) d\mathbb{P}(\textbf{g}_k) 
< \frac{1}{||A||^2}
\end{split}
\end{equation*}
which can be checked by numerically computing the integral for specific functions $p(.)$ and probability distributions $\mathbb{P}(\textbf{g}_k)$. If condition (\ref{stability_condition}) is satisfied for this suboptimal policy, then by Theorem \ref{stability_thm} the Kalman filter using the optimal relay configurations will also be stable. 
\end{example}

\subsection{Suboptimal relay configuration selection}
\label{suboptimal_relay_sec}
Lemma \ref{num_configs_lemma} has shown that the optimal way of choosing the relay configuration by checking each configuration is exponential in the number of relays $L$, which is computationally intensive when  $L$ is large. To reduce  computational complexity, 
a suboptimal method for determining a relay configuration is to optimize the operation of each relay $l$ independently of each other. The motivation for this method is that sometimes other relays may become unavailable, thus one should optimize the performance of each relay irrespective of what the other relays are doing. Specifically, consider subsets $I_l \subseteq \mathcal{I}_l$. Let $\bar{C}(I_l) = \textrm{col}(\{C_i, i \in I_l\})$, $\bar{R}_k(I_l) = \textrm{diag}(\{\breve{R}_{i,k}, i \in I_l\})$, and 
\begin{equation*}
%\label{fl_defn}
\begin{split}
\!f_{l,k}(X) & \triangleq\! A X A^T \!\!+\! Q \!-\! \sum_{I_l \subseteq \mathcal{I}_l} \mathbb{E}^l\Big[\prod_{i \in I_l} \Theta_{i,k} \prod_{j \notin I_l} (1- \Theta_{j,k}) | \textbf{g}_k, \phi_{l,k} \Big] 
\\ & \quad \times 
A X \bar{C}(I_l)^T \!\!\left(  \bar{C}(I_l) X \bar{C}(I_l)^T \!\!+\! \bar{R}_k(I_l) \right)^{\!\!-1} \!\!\!\bar{C}(I_l) X A^T 
\end{split}
\end{equation*}
where the terms $\mathbb{E}^l\Big[\prod_{i \in I_l} \Theta_{i,k} \prod_{j \notin I_l} (1- \Theta_{j,k}) | \textbf{g}_k, \phi_{l,k}  \Big]$ are computed assuming that relay $l$ is the only relay available. 
One then computes $f_{l,k}(P_k)$
for each of the operations $\phi_{l,k}$ that relay $l$ can perform, with the one that gives the smallest value of $\textrm{Tr}(f_{l,k}(P_k))$ then chosen. This optimization can be carried out for each relay independently of the other relays. 
The number of configurations that need to be checked at each time step $k$ is then 
$\sum_{l=1}^L \left(2^{M_l} -1 \right)$, which (compare to (\ref{num_configs}))  is no longer exponential in the number of relays $L$, and with $M_l$ usually being small as mentioned at the end of Section \ref{exhaustive_search_sec}. 
%We also note that this optimization can be formulated as a linear program similar to Section \ref{linear_programing_sec}; the details are omitted for brevity.

%\subsubsection{Random relay configuration selection}
%An even simpler scheme is to choose the relay configuration randomly at each time instant from a predetermined distribution (such as the uniform distribution), which is motivated by the random coding methods  proposed  for network coding, see e.g. \cite{HoLun}. Such a scheme could be attractive in large networks where optimizations may be difficult to perform. 

\section{A special case}
\label{special_case_sec}
Here we consider a special case where additional analytical results can be derived. We will study the effects of varying the packet reception probabilities and signal-to-noise ratios, which will provide some insight into the general behaviour.

Recall that in the case of two sensors and one relay, the relay can either i) forward $y_{1,k}$ if it is received, ii) forward $y_{2,k}$ if it is received, or iii) send $y_{1,k}^q \oplus y_{2,k}^q$ if both $y_{1,k}^q$ and $y_{2,k}^q$ are received. 
We will consider a scalar example with $\tilde{\lambda}_{1,k}=\rho_{1,k}^1=\rho_{2,k}^1=1, \forall k$, corresponding to the case where the channels from the relay to the gateway, and from the sensors to the relay, are error free. We will call $A=a$, $Q=q$, $C_1=c_1$, $C_2=c_2$, $R_1=r_1$, $R_2=r_2$, $\breve{R}_{1,k}=\breve{r}_{1,k}$, $\breve{R}_{2,k}=\breve{r}_{2,k}$. 
We have 
\begin{equation}
\label{EP_special_case}
\begin{split}
& \mathbb{E}[P_{k+1}|P_k,\textbf{g}_k,\bm{\phi}_k]  \\ &  \!\! = \left\{ \begin{array}{lcl}  \!\!a^2 P_k \!+\! q   \!-\! \frac{\lambda_{2,k} a^2 P_k^2}{P_k+\frac{1}{c_1^2/\breve{r}_{1,k}+c_2^2/\breve{r}_{2,k}}} \! -\!  \frac{(1-\lambda_{2,k}) a^2 P_k^2}{P_k + \frac{1}{c_1^2/\breve{r}_{1,k}}} & \!\!\!\!\!\!\!, & \!\!\!\!\! \textrm{forward } y_{1,k}^q \\
\!\!a^2 P_k \!+\! q \!-\!  \frac{\lambda_{1,k} a^2 P_k^2}{P_k+\frac{1}{c_1^2/\breve{r}_{1,k}+c_2^2/\breve{r}_{2,k}}} \!- \! \frac{(1-\lambda_{1,k}) a^2 P_k^2}{P_k + \frac{1}{c_2^2/\breve{r}_{2,k}}} & \!\!\!\!\!\!\!, & \!\!\!\!\!\textrm{forward } y_{2,k}^q \\
\!\! a^2 P_k \!+\! q \!-\!  \frac{(\lambda_{1,k} +\lambda_{2,k} -\lambda_{1,k} \lambda_{2,k})a^2 P_k^2}{P_k+\frac{1}{c_1^2/\breve{r}_{1,k}+c_2^2/\breve{r}_{2,k}}} & \!\!\!\!\!\!\!, & \!\!\!\!\!\textrm{send } y_{1,k}^q \!\oplus \!y_{2,k}^q
\end{array} \right.
\end{split}
\end{equation}
We want to see under what conditions the XOR operation $y_{1,k}^q\oplus y_{2,k}^q$ outperforms forwarding of measurements. First let the  terms $c_1^2/\breve{r}_{1,k}$ and $c_2^2/\breve{r}_{2,k}$, which can be regarded as the signal-to-noise ratios (SNRs) of sensors 1 and 2, be fixed. From (\ref{EP_special_case}) we see that sending $y_{1,k}^q \oplus y_{2,k}^q$ is better than forwarding $y_{1,k}^q$ when 
\begin{equation}
\label{lambda1_condition}
\frac{\lambda_{1,k}(1-\lambda_{2,k})}{P_k \!+\! \frac{1}{c_1^2/\breve{r}_{1,k}+c_2^2/\breve{r}_{2,k}}}\! >\! \frac{1-\lambda_{2,k}}{P_k \!+\! \frac{1}{c_1^2/\breve{r}_{1,k}} } \textrm{ or } \lambda_{1,k} \!>\! \frac{P_k \!+\! \frac{1}{c_1^2/\breve{r}_{1,k}+c_2^2/\breve{r}_{2,k}}}{P_k + \frac{1}{c_1^2/\breve{r}_{1,k}} }
\end{equation}
Similarly,   sending $y_{1,k}^q \oplus y_{2,k}^q$ is better than forwarding $y_{2,k}^q$ when 
\begin{equation}
\label{lambda2_condition}
\lambda_{2,k} > \frac{P_k + \frac{1}{c_1^2/\breve{r}_{1,k}+c_2^2/\breve{r}_{2,k}}}{P_k + \frac{1}{c_2^2/\breve{r}_{2,k}} }
\end{equation}
From (\ref{lambda1_condition})-(\ref{lambda2_condition}), sending $y_{1,k}^q \oplus y_{2,k}^q$ is best when both packet reception probabilities $\lambda_{1,k}$ and $\lambda_{2,k}$ are above certain thresholds, which in turn implies that the instantaneous channel gains $g_{1,k}$ and $g_{2,k}$ need to be above some thresholds. Thus, for lower quality channels, forwarding of measurements gives better performance than the network coding operation. 
The intuitive explanation for this is that when the gateway receives $y_{1,k}^q \oplus y_{2,k}^q$, it needs one other measurement (either $y_{1,k}^q$ or $y_{2,k}^q$)  in order to be useful. In contrast, if the relay forwards $y_{1,k}^q$ or $y_{2,k}^q$, this value (if received at the gateway) can be used even if the direct transmissions from the sensors are lost.
 
Alternatively, regard $\lambda_{1,k}$ and $\lambda_{2,k}$ as being fixed. Rewriting (\ref{lambda1_condition}) and (\ref{lambda2_condition}),  sending $y_{1,k}^q \oplus y_{2,k}^q$ is better than forwarding $y_{1,k}^q$ when
\begin{equation}
\label{SNR_condition1}
\frac{\lambda_{1,k}}{c_1^2/\breve{r}_{1,k}} - \frac{1}{c_1^2/\breve{r}_{1,k}+c_2^2/\breve{r}_{2,k}} > P_k(1-\lambda_{1,k}),
\end{equation}
and sending $y_{1,k}^q \oplus y_{2,k}^q$ is better than forwarding $y_{2,k}^q$ when
\begin{equation}
\label{SNR_condition2}
\frac{\lambda_{2,k}}{c_2^2/\breve{r}_{2,k}} - \frac{1}{c_1^2/\breve{r}_{1,k}+c_2^2/\breve{r}_{2,k}} > P_k(1-\lambda_{2,k})
\end{equation}
For fixed $\lambda_{1,k}$ and $\lambda_{2,k}$, we see that if either $c_1^2/\breve{r}_{1,k}$ or $c_2^2/\breve{r}_{2,k}$ is sufficiently large, then conditions (\ref{SNR_condition1}) and (\ref{SNR_condition2}) cannot both be simultaneously satisfied. 
%To show that conditions (\ref{SNR_condition1}) and (\ref{SNR_condition2}) can actually be satisfied for certain values of $c_1^2/r_1$ and $c_2^2/r_2$, suppose that $c_1^2/r_1=c_2^2/r_2$. Then conditions (\ref{SNR_condition1}) and (\ref{SNR_condition2}) become 
%$$\frac{\lambda_{1,k}-1/2}{c_1^2/r_{1}} > P_k(1-\lambda_{1,k})$$
%and 
%$$\frac{\lambda_{1,k}-1/2}{c_1^2/r_{1}} > P_k(1-\lambda_{2,k}),$$
%which can both be satisfied if $c_1^2/r_1=c_2^2/r_2$ is sufficiently small, provided $\lambda_{1,k} > 1/2$ and $\lambda_{2,k} > 1/2$. 
The intuitive reason for requiring the signal-to-noise ratios to be small in order for network coding to give benefits, is that the relative performance gains by having both measurements available at the gateway (vs. just one of the measurements) is greater at low SNRs than high SNRs. 

These qualitative observations, that the XOR operation needs channel conditions to be good,  or for the signal-to-noise ratios to be low, in order to give benefits over forwarding of measurements, have been observed in simulations for more general cases of packet dropping links to and from the relays, and with larger networks. Similar behaviour has also been reported in the network coding literature \cite{RajalinghamHoLeNgoc}. 

\section{Relay configuration selection and power control}
\label{power_control_sec}
In Section \ref{optimal_relay_sec} the sensor and relay transmit powers were assumed to be fixed. However, the presence of time-varying fading channels will also allow for the use of power control techniques to further improve performance. In this section we present one possible formulation 
which optimizes the estimation performance subject to a sum of transmit powers constraint. 

As in Section \ref{optimal_relay_sec}, we assume that full channel state information (CSI) is available at the receiver, with $\textbf{g}_k$ in (\ref{gk_vector})
representing the set of all channel gains at time $k$.
 The transmit powers of the sensors and relays can then depend on both the instantaneous channel gains $\textbf{g}_k$ and the error covariance $P_k$, with these transmit powers being computed at the gateway (which is assumed to have more computational resources than the sensors and relays) and fed back to the sensors and relays before transmission occurs at the next time step. 
 Denote $\textbf{u}_k(\textbf{g}_k,P_k)  \triangleq \{u_{1,k}(\textbf{g}_k,P_k),\dots,u_{M,k}(\textbf{g}_k,P_k),
 \tilde{u}_{1,k}(\textbf{g}_k,P_k),\dots,\tilde{u}_{L,k}(\textbf{g}_k,P_k)\}$
%\begin{equation*}
%\begin{split}
% \textbf{u}_k(\textbf{g}_k,P_k) & \triangleq \{u_{1,k}(\textbf{g}_k,P_k),\dots,u_{M,k}(\textbf{g}_k,P_k),%\\ & \quad\quad 
% \tilde{u}_{1,k}(\textbf{g}_k,P_k),\dots,\tilde{u}_{L,k}(\textbf{g}_k,P_k)\}
%\end{split}
%\end{equation*}
as the set of all transmit powers at time $k$. 

\subsection{Optimal power control for a given relay configuration}
For a given relay configuration, we pose the following power control problem:
\begin{equation}
\label{one_step_problem}
\begin{split}
& \min_{ \textbf{u}_k(\textbf{g}_k,P_k) }  \textrm{Tr} \mathbb{E}[P_{k+1} | P_k, \textbf{g}_k, \bm{\phi}_k ] 
\\ & \textrm{s.t. }
 \sum_{i=1}^M u_{i,k}(\textbf{g}_k,P_k)   + \sum_{l=1}^L \tilde{u}_{l,k}(\textbf{g}_k,P_k) \leq u_{\textrm{tot}}
\end{split}
\end{equation}
which minimizes the expected one-step ahead error covariance subject to the sum power $\sum_{i=1}^M u_{i,k}(\textbf{g}_k,P_k)  + \sum_{l=1}^L \tilde{u}_{l,k}(\textbf{g}_k,P_k)$ being less than $u_{\textrm{tot}}$. 
Due to the objective  being a complicated nonlinear function of the transmit powers $\textbf{u}_k$, the optimization problem (\ref{one_step_problem}) is in general non-convex and will need to be solved numerically using global optimization algorithms \cite{LocatelliSchoen}. 

\subsection{Joint relay configuration selection and power control}
Problem (\ref{one_step_problem}) is for a given relay configuration. To optimally choose both the relay configuration and transmission powers,  we can in principle solve $\prod_{l=1}^L \left( 2^{M_l}-1 \right) $ instances of problem (\ref{one_step_problem}) at each time step (for each of the configurations, see Lemma \ref{num_configs_lemma}), and choose the relay configuration that gives the smallest value for the objective function, which however is very computationally intensive. 

A less computationally intensive suboptimal scheme is to first choose a relay configuration by assuming a simple power allocation (e.g. that the total power $u_{\textrm{tot}}$ is equally divided between the sensors and relays), and using the suboptimal method of Section \ref{suboptimal_relay_sec} to choose a relay configuration. For this chosen relay configuration, we then further optimize the transmission powers by solving the power control problem (\ref{one_step_problem}).

\section{Numerical studies}
\label{numerical_sec}
We first look at the performance differences between the optimal relay configuration selection and the suboptimal methods of Section \ref{optimal_relay_sec}. We consider a situation with two sensors and two relays, where each of the relays can listen to both sensor transmissions, see Fig. \ref{system_model_two_relay_new}.
 We consider the scalar case with $a=0.95$, $q=1$, $c_1=c_2=1$, $r_1=r_2=1$. For simplicity, we assume that the links from the sensors to the relays are perfect, with the  fading channels (from the sensors to gateway, and from the relays to gatewary) being exponentially distributed with mean 1, which models the case of Rayleigh fading \cite{Proakis}. Similar to \cite{Quevedo_automatica}, we assume that the digital communication uses BPSK transmission  with $b=6$ bits per packet, so that the function $p(.)$ in Section \ref{model_sec} has the form (\ref{f_BPSK}). We distribute the transmit powers equally between the sensors and relays. Fig. \ref{two_relay_time_varying_plot} plots the average sum power and expected error covariance $\mathbb{E}[P_k]$ (obtained by time averaging $(x_{k}-\hat{x}_{k|k-1})(x_{k}-\hat{x}_{k|k-1})^T$ over 10000 Monte Carlo iterations), for the optimal and suboptimal relay configuration selection methods. For comparison we also plot the performance for  the cases of: 1) no relay, 2) a scheme where the relay always performs the XOR operation \cite{QuevedoOstergaardAhlen}, 3) a scheme where the relay sends the two sensor measurements with less accuracy by removing half the bits\footnote{Here the relay removes from each quantized sensor measurement the least significant $b/2$ bits. Then in the Kalman filter equations $(\ref{KF_eqns})$, for the case where the direct transmissions are not successful, in the expression for  $\breve{R}_k$ we replace $\delta_b = \frac{4 b \ln 2}{3 \times 2^{2b}}$ with $\delta_{b/2}=  \frac{2 b \ln 2}{3 \times 2^{b}}$ whenever we use a measurement where half the bits have been removed.}, and 4) a scheme where the gateway can ask for each lost transmission to be retransmitted once\footnote{Here we assume that additional transmit power (same as the power for a
single transmission) is used for each retransmission, with a successfully
retransmitted measurement (from time $k$) available to the Kalman filter at time
$k+1$, which now utilizes a buffer similar to \cite{Schenato}.}.  In each case, the expected error covariance decreases as the average power is increase. Since by (\ref{f_BPSK}) larger powers imply higher packet reception probabilities, this behaviour is in agreement with Lemma \ref{lambda_monotonicity_lemma}. We also see that the suboptimal method that optimizes each relay separately gives very close performance to the optimal method, and significantly outperforms the other schemes. 
 %Note that in the case of two sensors and two relays, for optimal relay configuration selection one needs to compare between nine different configurations at each time step, whereas for the suboptimal scheme each relay will need to compare between three different configurations. 
 \begin{figure}[t!]
\centering 
\includegraphics[scale=0.5]{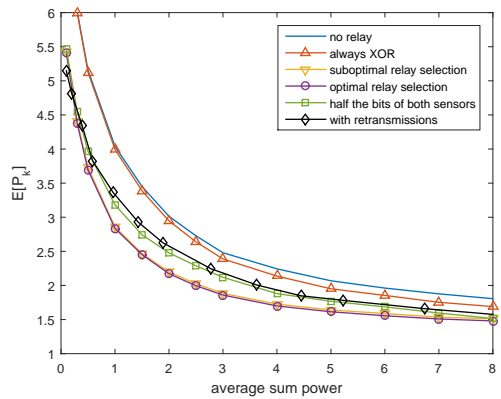} 
\caption{Optimal and suboptimal relay configuration selection}
\label{two_relay_time_varying_plot}
\end{figure} 

We next consider the case of two sensors and one relay, with Rayleigh fading for each of the fading channels. We choose $g_{1,k}$, $g_{2,k}$ to have mean $1$, while $\tilde{g}_{1,k}$, $h_{1,k}$, $h_{2,k}$ have mean $4$. This models the case where power decays in free space as $1/d^2$ where $d$ is the distance from the transmitter \cite{Proakis}, with the relay located approximately halfway between the sensors and gateway. 
%The parameters used are $a=1.2$, $q=1$, $c_1^2/r_1=c_2^2/r_2=1$.  We again use  binary phase shift keying (BPSK) transmission  with $b=4$ bits per packet.
 \begin{figure}[t!]
\centering 
\includegraphics[scale=0.5]{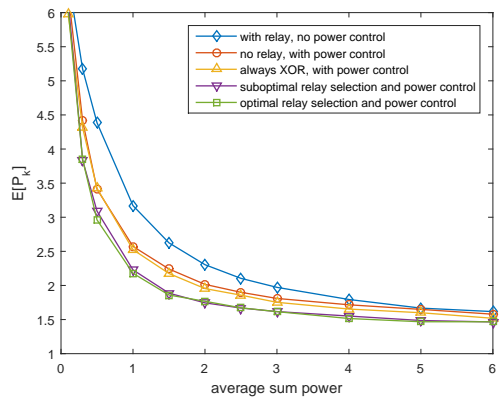} 
\caption{Power control and relay configuration selection}
\label{one_relay_time_varying_power_plot}
\end{figure} 
In Fig. \ref{one_relay_time_varying_power_plot} we plot $\mathbb{E}[P_k]$ (obtained by time averaging $(x_{k}-\hat{x}_{k|k-1})(x_{k}-\hat{x}_{k|k-1})^T$ over 10000 Monte Carlo iterations) for different sum powers, obtained by solving problem (\ref{one_step_problem})  using the  \verb=fmincon= routine in Matlab for each relay configuration and selecting the best one. We also plot the performance of the suboptimal scheme where a relay configuration is first chosen (assuming equal power allocation) and then power control is performed. 
We compare this with the case where there is no power control, with the sensors and relay using the same transmit power at all times, but with the best relay configuration chosen at each time step. 
Additionally, we plot the case where the relay always performs the XOR operation, and the case without a relay but with power control.  
We see that doing power control gives significant benefits,  with the best performance when one optimizes both the relay configuration and transmit powers. The suboptimal scheme where a relay configuration is first chosen by assuming equal power allocation, and then the powers are optimized, performs very close to the optimal scheme.  Comparing the plots using power control, we see that for a given expected error covariance the average  power required is signficantly less (at least 30-40$\%$) when a relay is used.

\section{Conclusion}
\label{conclusion_sec}
This note has studied the use of relays for Kalman filtering with multiple sensors over packet dropping links, where the packet reception probabilities are governed by fading channel gains and sensor and relay transmit powers. By allowing relays to either forward one of the sensor's measurements or perform a network coding operation, we have considered the problem of determining the optimal relay configuration at each time step, together with a simpler suboptimal method. We have also studied the use of power control in addition to selecting the best relay configuration, to further improve performance. Numerical results have demonstrated that the use of relays can lead to power savings of at least 30-40$\%$. Future work will include studies of other strategies that relays can perform, and extending our setup to  utilize relays in multi-hop networks \cite{DolzQuevedo_IFAC}.

\bibliography{IEEEabrv,kalmanrelay}
\bibliographystyle{IEEEtran} 

\end{document}